\documentclass[conference]{IEEEtran}
\usepackage{amsthm, amssymb, amsmath}
\usepackage{verbatim, float}
\usepackage{mathrsfs}
\usepackage{graphics}
\usepackage[usenames,dvipsnames]{pstricks}
\usepackage{epsfig}
\usepackage{pst-grad} % For gradients
\usepackage{pst-plot} % For axes
\usepackage{cases}
\usepackage{algorithmic}
\usepackage{bm}

\usepackage[top=1.4cm, bottom=1.4cm, left=1.33cm, right=1.33cm]{geometry}
\usepackage{url}

\usepackage{hhline}

\usepackage{etoolbox}
\makeatletter
\patchcmd{\@makecaption}
  {\scshape}
  {}
  {}
  {}
\makeatother
\usepackage[singlelinecheck=false]{caption}
\usepackage{diagbox}

\theoremstyle{plain}
\newtheorem{theorem}{Theorem}
\newtheorem{lemma}{Lemma}

\newtheorem{corollary}{Corollary}

\theoremstyle{definition}
\newtheorem{definition}{Definition}
\newtheorem{example}{Example}
\newtheorem{remark}{Remark}

\newcommand{\C}{{\mathcal C}}

\newcommand{\ba}{{\boldsymbol a}}
\newcommand{\bb}{{\boldsymbol b}}

\newcommand{\ff}{\mathbb{F}}

\newcommand{\tr}{\mathsf{Tr}}

\newcommand{\define}{\stackrel{\mbox{\tiny $\triangle$}}{=}}

\newcommand{\Cd}{\mathcal{C}^\perp}
\newcommand{\rsk}{\text{RS}(A,k)}
\newcommand{\grskl}{\text{GRS}(A,k,\boldsymbol{\lambda})}
\newcommand{\grsnkl}{\text{GRS}(A,n-k,\boldsymbol{\lambda})}

\newcommand{\fa}{f(\alpha)}
\newcommand{\fas}{f(\alpha^*)}
\newcommand{\fab}{f(\overline{\alpha})}

\newcommand{\pa}{p(\alpha)}
\newcommand{\pas}{p(\alpha^*)}

\newcommand{\pia}{p_i(\alpha)}
\newcommand{\pias}{p_i(\alpha^*)}
\newcommand{\piab}{p_i(\overline{\alpha})}

\newcommand{\qia}{q_i(\alpha)}
\newcommand{\qias}{q_i(\alpha^*)}
\newcommand{\qiab}{q_i(\overline{\alpha})}

\newcommand{\poas}{p_1(\alpha^*)}

\newcommand{\qoab}{q_1(\overline{\alpha})}

\newcommand{\ptmoas}{p_{t-1}(\alpha^*)}

\newcommand{\pta}{p_t(\alpha)}
\newcommand{\ptas}{p_t(\alpha^*)}
\newcommand{\ptab}{p_t(\overline{\alpha})}

\newcommand{\qta}{q_t(\alpha)}
\newcommand{\qtas}{q_t(\alpha^*)}
\newcommand{\qtab}{q_t(\overline{\alpha})}

\newcommand{\qsbz}{Q_{\alpha^*,\overline{\alpha}}(z)}

\newcommand{\Qab}{Q_{\alpha, \beta}}
\newcommand{\Qabz}{Q_{\alpha, \beta}(z)}
\newcommand{\Qba}{Q_{\beta,\alpha}}
\newcommand{\Qbaz}{Q_{\beta,\alpha}(z)}

\renewcommand{\psi}{p^*_i}
\newcommand{\psix}{p^*_i(x)}
\newcommand{\psia}{p^*_i(\alpha)}
\newcommand{\psias}{p^*_i(\alpha^*)}

\newcommand{\psta}{p^*_t(\alpha)}
\newcommand{\pstas}{p^*_t(\alpha^*)}
\newcommand{\pstab}{p^*_t(\overline{\alpha})}
\newcommand{\psoas}{p^*_1(\alpha^*)}

\newcommand{\psoab}{p^*_1(\overline{\alpha})}

\newcommand{\pstmoab}{p^*_{t-1}(\overline{\alpha})}

\newcommand{\Kab}{K_{\alpha,\beta}}
\newcommand{\Kba}{K_{\beta,\alpha}}

\newcommand{\Ksb}{K_{\alpha^*,\overline{\alpha}}}

\begin{document}

\title{Repairing Reed-Solomon Codes\\ With Two Erasures
\thanks{
H. Dau and O. Milenkovic are with the Coordinated Science Laboratory, University of Illinois at Urbana-Champaign, 1308 W. Main Street, Urbana, IL 61801, USA. Emails: \{hoangdau, milenkov\}@illinois.edu.
I. Duursma is with the Departments of Mathematics, and also with the Coordinated Science Laboratory, University of Illinois at Urbana-Champaign, 1409 W. Green St, Urbana, IL 61801, USA. Email: duursma@illinois.edu.
H. M. Kiah is with the Division of Mathematical Sciences, School of Physical and Mathematical Sciences, Nanyang Technological University, 21 Nanyang Link, Singapore 637371. Email: hmkiah@ntu.edu.sg.} 
}
\author{Hoang Dau, Iwan Duursma, Han Mao Kiah, and Olgica Milenkovic}
\author{
  \IEEEauthorblockN{
    Hoang~Dau\IEEEauthorrefmark{1},
		Iwan Duursma\IEEEauthorrefmark{2},
		Han Mao Kiah\IEEEauthorrefmark{3},
    and Olgica~Milenkovic\IEEEauthorrefmark{4}
    }
  {\normalsize
    \begin{tabular}{ccc}
      \IEEEauthorrefmark{1}\IEEEauthorrefmark{2}\IEEEauthorrefmark{4}University of Illinois at Urbana-Champaign, 
			\IEEEauthorrefmark{3}Nanyang Technological University \\
			Emails: \{\IEEEauthorrefmark{1}hoangdau, \IEEEauthorrefmark{2}duursma, \IEEEauthorrefmark{4}milenkov\}@illinois.edu, \IEEEauthorrefmark{3}hmkiah@ntu.edu.sg
    \end{tabular}}\vspace{-3ex}
    }
\date{}
\maketitle
%\pagestyle{empty}
%%%%%%%%%%%%%%%%%%%%%%%%%%%%%%%%%%

\begin{abstract}
Despite their exceptional error-correcting properties, Reed-Solomon (RS) codes have been overlooked in distributed storage applications 
due to the common belief that they have poor repair bandwidth: A naive repair approach would require the whole file to be reconstructed in 
order to recover a single erased codeword symbol. In a recent work, Guruswami and Wootters (STOC'16) proposed a single-erasure 
repair method for RS codes that achieves the optimal repair bandwidth amongst all linear encoding schemes. 
%Their key idea is to recover the erased symbol by collecting a sufficient number of its traces, each of which can be constructed from a number of traces of other symbols. As all traces belong to a subfield of the defining field of the RS code and many of them are linearly dependent, the total repair bandwidth is significantly reduced compared to that of the naive repair scheme. 
We extend their trace collection technique to cope with two erasures.     
\end{abstract}

\section{Introduction}
\label{sec:intro}

\subsection{Background}

The \emph{repair bandwidth} is an important performance metric of erasure codes in the context 
of distributed storage~\cite{Dimakis_etal2010}.
In such a system, for a chosen field $F$, a data vector in $F^k$ is mapped to a codeword vector in $F^n$, whose entries are 
stored at different storage nodes. When a node fails, the symbol stored at that node is erased (lost). 
A replacement node (RN) has to recover the content
stored at the failed node by downloading information from the other nodes. 
The repair bandwidth is the total amount of information that the RN
has to download in order to successfully complete the repair process. 

Reed-Solomon (RS) codes~\cite{ReedSolomon1960}, which have been extensively studied in theory~\cite{MW_S} and widely used in practice, were believed to have prohibitively high repair bandwidth. In a naive repair scheme, recovering the content stored at a \emph{single} failed node would require downloading the \emph{whole} file, i.e., $k$ symbols over $F$. The poor performance in repairing failed nodes of RS codes motivated the introduction of repair-efficient codes such as regenerating codes~\cite{Dimakis_etal2010} and locally repairable codes~\cite{OggierDatta2011,GopalanHuangSimitciYekhanin2012,PapailiopoulosDimakis2012}.

\begin{figure}[t]
\centering
\includegraphics[scale=0.8]{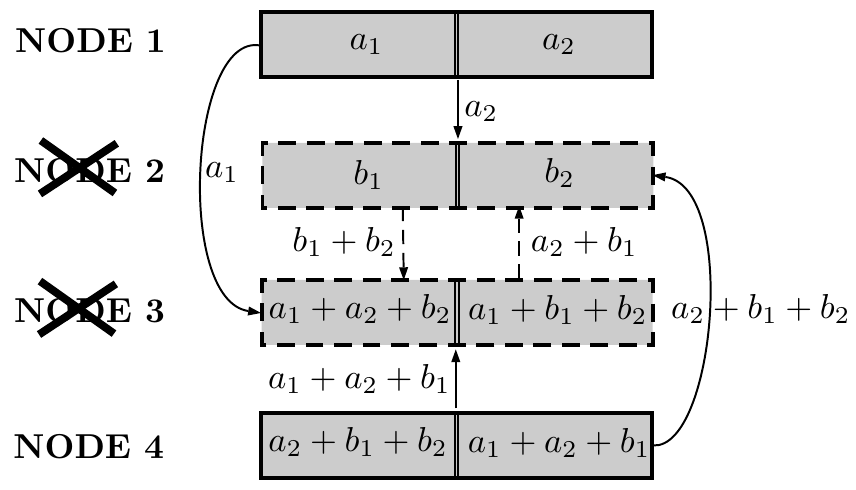}
\caption{A toy example illustrating the repair procedure for \emph{two} failed nodes in a four-node storage system
based on a $[4,2]$ Reed-Solomon code over $\ff_4$. 
%The procedure in question is the distributed repair scheme introduced in Section~\ref{subsec:distributedII}. 
The stored file is $\big((a_1,a_2),(b_1,b_2)\big) \in \ff_4^2$, where $a_1$, $a_2$, $b_1$, and $b_2$ are bits in $\ff_2$. 
Suppose that Node~2 and Node~3 fail
simultaneously. In the Download Phase, each replacement node first downloads two bits (along the solid arrows) from the two available nodes, namely Node~1 and Node~4. In the Collaboration Phase, the replacement nodes communicate with each other to complete their own repair processes by exchanging two extra bits (along the dashed arrows), computed based on the previously downloaded bits. 
%Note that each of the extra bits can be computed as a linear combination of bits that the corresponding RN downloaded from the available nodes in the earlier step of repair. The content of each failed node is then reconstructed from three downloaded bits. 
}
\label{fig:toy_example}
\vspace{-10pt}
\end{figure}

Guruswami and Wootters~\cite{GuruswamiWootters2016} recently proposed a bandwidth-optimal linear repair method based 
on RS codes. The key idea behind their method is to recover a single erased symbol by collecting a sufficiently large number of
its (field) traces, each of which can be constructed from a number of traces of other 
symbols. As all traces belong to a subfield $B$ of $F$ and traces from 
the same symbol are related, the total repair bandwidth can be significantly reduced.
The repair scheme obtained by Guruswami and Wootters~\cite{GuruswamiWootters2016}, however, only applies to the case
of one erasure, or in other words, one failed node. 

\subsection{Our Contribution}

We propose an extension of the Guruswami-Wootters repair scheme that can ensure recovery from two erasures. We provide two \emph{distributed} schemes for Reed-Solomon codes, both of which use
the same repair bandwidth \emph{per erasure} as in the case of a single erasure. 
In these repair schemes, the two RNs first download repair data
from all available nodes (Download Phase). They subsequently collaborate to exchange the data in order to complete the repair process at each node (Collaboration Phase).
The first scheme has a \emph{collaboration depth one}, that is, in the Collaboration Phase, the two RNs send out repair data to each other simultaneously in one round. This scheme works whenever the field extension degree $t$ is divisible by the characteristic of the field $F$. An example illustrating the first scheme is given in Fig.~\ref{fig:toy_example}. The second scheme has a \emph{collaboration depth two}, that is, in the Collaboration Phase, one RN receives the repair data from the
other RN, completes its repair process, and then sends out
its repair data to the other node. This scheme applies to all field extension degrees. 

\subsection{Organization}    

The paper is organized as follows. We first provide relevant definitions and introduce the terminology used throughout the paper. We then proceed to 
discuss the Guruswami-Wootters repair scheme for RS codes in the presence of a single erasure in Section~\ref{sec:GW}. Our main results -- repair schemes for RS codes in the presence of two erasures -- are presented in Section~\ref{sec:two_erasures}.
For a thorough literature review on the related works on \emph{cooperative regenerating codes} and the motivation for 
repairing multiple erasures, the interested reader is referred to our
companion paper~\cite{DauDuursmaKiahMilenkovic2016}.

\section{Repairing One Erasure in Reed-Solomon Codes}
\label{sec:GW}

We start by introducing relevant definitions and the notation used in all subsequent derivations, and then proceed to 
review the approach proposed by Guruswami and Wootters~\cite{GuruswamiWootters2016} for repairing a single erasure/node failure in RS codes. 

\subsection{Definitions and Notations}
\label{subsec:def_not}

Let $[n]$ denote the set $\{1,2,\ldots,n\}$. Let $B = \text{GF}(p^m)$ be the finite field of $p^m$ elements, for some prime $p$
and $m \geq 1$. Let $F = \text{GF}(p^{mt})$ be a field extension of $B$, where $t \geq 1$. 
We often refer to the elements of $F$ as \emph{symbols} and the elements of $B$ as \emph{sub-symbols}. 
We can also treat $F$ as a vector space of dimension $t$ over $B$, i.e. $F \cong B^t$, and hence each symbol in $F$ may be represented 
as a vector of length $t$ over $B$. A linear $[n,k]$ code $\C$ over $F$ is a subspace of $F^n$ of dimension $k$. Each element of a code is referred to as a codeword. 
The dual of a code $\C$, denoted $\Cd$, is the orthogonal complement of $\C$. 
 
\begin{definition} 
\label{def:RS}
Let $F[x]$ denote the ring of polynomials over $F$. The Reed-Solomon code $\rsk \subseteq F^n$ of dimension $k$ over a finite
field $F$ with evaluation points $A=\{\alpha_1,\alpha_2,\ldots, \alpha_n\}\subseteq F$
is defined as: 
\[
\rsk = \Big\{\big(f(\alpha_1),\ldots,f(\alpha_n)\big) \colon f \in F[x], \deg(f) < k \Big\}.
\]
\end{definition} 
A \emph{generalized} Reed-Solomon code, $\grskl$, where $\boldsymbol{\lambda} = (\lambda_1,\ldots,\lambda_n)\in F^n$, is defined similarly to a Reed-Solomon code, except that the codeword
corresponding to a polynomial $f$ is now defined as $\big( \lambda_1f(\alpha_1),\ldots,\lambda_n f(\alpha_n) \big)$, $\lambda_i \neq 0$ for all $i \in [n]$. 
It is well known that the dual of an RS code $\rsk$, for any $n \leq |F|$, is a generalized RS code $\grsnkl$, for some multiplier vector $\boldsymbol{\lambda}$~(see~\cite[Chp.~10]{MW_S}). Whenever clear from the context, we use $f(x)$ to denote a polynomial of degree at most $k-1$, which corresponds to a codeword of the RS code $\C=\rsk$, and $p(x)$ to denote a polynomial of degree at most $n-k-1$, which corresponds to a dual codeword in $\Cd$. Since
$
\sum_{\alpha \in A}\pa(\lambda_{\alpha}\fa) = 0, 
$
we refer to such a polynomial $p(x)$ as a check polynomial for $\C$. 
Note that when $n = |F|$, we have $\lambda_\alpha = 1$ for all $\alpha \in F$. 
In general, as recovering $\fa$ is equivalent to recovering $\lambda_{\alpha}\fa$, to simplify the notation, we often omit the factor $\lambda_\alpha$ in the equation above.

\subsection{The Guruswami-Wootters Repair Scheme for One Erasure}
\label{subsec:GW}

Suppose that the polynomial $f(x) \in F[x]$ corresponds to a codeword in the RS code $\C=\rsk$ and that $f(\alpha^*)$ is the erased 
symbol, where $\alpha^* \in A$ is an evaluation point of the code. 

Given that $F$ is a field extension of $B$ of degree $t$, i.e. $F = \text{GF}(p^{mt})$ and $B = \text{GF}(p^m)$, for 
some prime $p$, one may define the field trace of any symbol $\alpha \in F$ as
$\mathsf{Tr}_{F/B}(\alpha) = \sum_{i = 0}^{t-1} \alpha^{|B|^i}$,
which is always a sub-symbol in $B$. We often omit the subscript $F/B$ for succinctness. 
The key points in the repair scheme proposed by Guruswami and Wootters~\cite{GuruswamiWootters2016} can be 
summarized as follows. Firstly, each symbol in $F$ can be recovered from its $t$ independent traces. More precisely, given a 
	basis $u_1,u_2,\ldots,u_t$ of $F$ over $B$, any $\alpha \in F$ can be uniquely determined given the values of $\tr(u_i\,\alpha)$ for $i\in [t]$, i.e. $\alpha = \sum_{i=1}^t\tr(u_i \alpha)u^\perp_i$, where $\{u^\perp_i\}_{i=1}^t$ is the dual (trace-orthogonal) basis of $\{u_i\}_{i=1}^t$ (see, for instance~\cite[Ch.~2, Def.~2.30]{LidlNiederreiter1986}).
Secondly, when $n-k \geq |B|^{t-1}$, the trace function also provide checks that generate repair equations whose coefficients are linearly dependent over $B$, which keeps the repair cost low. 

Note that the checks of $\C$ are precisely those polynomials $p(x) \in F[x]$ with $\deg(p) < n - k$. 
It turns out that for $n-k \geq |B|^{t-1}$, we can define checks that take part in the repair process via the trace function described above. For each $u \in F$
and $\alpha \in F$, we define the polynomial
\begin{equation} 
\label{eq:p}
p_{u,\alpha}(x) = \tr\big(u(x-\alpha)\big)/(x-\alpha).
\end{equation} 
By the definition of a trace function, the following lemma follows in a straightforward manner. 

\begin{lemma}[\cite{GuruswamiWootters2016}]
\label{lem:p}
The polynomial $p_{u,\alpha}(x)$ defined in~\eqref{eq:p} satisfies the following properties.\\
\quad (a) $\deg(p_{u,\alpha}) = |B|^{t-1}-1$;\quad \text{(b)} $p_{u,\alpha}(\alpha) = u$.
\end{lemma}

By Lemma~\ref{lem:p}~(a), $\deg(p_{u,\alpha}) = |B|^{t-1}-1 < n - k$. Therefore, the polynomial $p_{u,\alpha}(x)$ corresponds to a codeword of $\Cd$ and is a check for $\C$. 
Now let $U = \{u_1,\ldots,u_t\}$ be a basis of $F$ over $B$, and set \vspace{-5pt}
\[
p_i(x) \define p_{u_i,\alpha^*}(x) = \tr\big(u_i(x-\alpha^*)\big)/(x-\alpha^*), \quad i \in [t].
\]  
These $t$ polynomials correspond to $t$ codewords of $\Cd$. 
Therefore, we obtain $t$ equations of the form \vspace{-5pt}
\begin{equation} 
\label{eq:GW_repair}
\pias\fas = - \sum_{\alpha \in A \setminus \{\alpha^*\}} \pia\fa, \quad i \in [t]. 
\end{equation} 
A key step in the Guruswami-Wootters repair scheme is to apply the trace function to
both sides of \eqref{eq:GW_repair} to obtain $t$ different \emph{repair equations} \vspace{-5pt}
\begin{equation} 
\label{eq:GW_repair_trace}
\tr\big(\pias\fas\big) = - \sum_{\alpha \in A \setminus \{\alpha^*\}} \tr\big(\pia\fa\big), \ i \in [t]. 
\end{equation}
According to Lemma~\ref{lem:p}~(b), $\pias = u_i$, for $i = 1,\ldots,t$. 
%Moreover, by the linearity of the trace function, we also have
%\[
%\tr\big(\pia\fa\big) = \tr\big(u_i(\alpha -\alpha^*)\big) \times \tr\Big(\dfrac{f(\alpha)}{\alpha-\alpha^*}\Big). 
%\]
%Therefore, one can rewrite~\eqref{eq:GW_repair_trace} as follows, for $i = 1,\ldots,t$, 
Moreover, by the linearity of the trace function, we can rewrite~\eqref{eq:GW_repair_trace} as follows. For $i = 1,\ldots,t$, \vspace{-5pt}
\begin{equation} 
\label{eq:GW_repair_trace_explicit}
\tr\big(u_i\fas\big) = - \sum_{\alpha \in A \setminus \{\alpha^*\}} \tr\big(u_i(\alpha -\alpha^*)\big) \times \tr\Big(\dfrac{f(\alpha)}{\alpha-\alpha^*}\Big). 
\end{equation}
The right-hand side sums of the equations~\eqref{eq:GW_repair_trace_explicit} can be computed by downloading the repair trace 
$\tr\Big(\frac{f(\alpha)}{\alpha-\alpha^*}\Big)$ from the node storing $\fa$, for each
$\alpha \in A \setminus \{\alpha^*\}$.
As a consequence, the $t$ independent traces $\tr\big(u_i\fas\big)$, $i = 1,\ldots,t$, of $\fas$ can be determined by downloading one sub-symbol from each of the $n-1$ available nodes. The erased symbol $\fas$ can subsequently be recovered from its $t$ independent traces.
By~\cite[Cor.~1]{DauMilenkovic2017}, this scheme is bandwidth-optimal when $n = |F|$ and $k = n(1-1/|B|)$.

\section{Repairing Two Erasures in Reed-Solomon Codes}
\label{sec:two_erasures}

We consider the same setting as in Section~\ref{subsec:GW}, i.e. $n - k \geq |B|^{t-1}$, where $B = \text{GF}(p^m)$ and $F = \text{GF}(p^{mt})$, 
and assume that $\C$ is an RS code $\rsk$ over $F$. 
However, we now suppose that two codeword symbols, say $\fas$ and $\fab$, 
are erased. 
Two repair schemes are proposed, both of which use the same bandwidth per erasure as in the case of a single erasure in~\cite{GuruswamiWootters2016}. 
%We propose two different repair schemes to recover these symbols, both of which use the same bandwidth per erasure as in the case of a single erasure.  

\subsection{General Idea}
\label{subsec:general_idea}

We first discuss the challenges associated with repairing two erased symbols and then proceed to 
describe our strategy for dealing with this repair scenario.
A check $p(x)$ is said to \emph{involve} a codeword symbol $\fa$ if $\pa \neq 0$. 
When only one symbol $\fas$ is erased, every check $p(x)$ that involves $\fas$ can be used to generate a repair equation as follows.
\begin{equation} 
\label{eq:trace_repair}
\tr\big(\pas \fas\big) = -\sum_{\alpha \in A \setminus \{\alpha^*\}}\tr\big(\pa\fa\big). 
\end{equation}
However, when two symbols $\fas$ and $\fab$ are erased, in order to, say, recover $\fas$, we no longer have
the freedom to use every possible check that involves $\fas$. Indeed, those checks that involve both $\fas$ and $\fab$
cannot be used in a straightforward manner for repair, because we cannot simply compute the right-hand side sum of \eqref{eq:trace_repair} without retrieving some information from $\fab$.

The gist of our approach is to first generate those checks that only involve one codeword symbol, $\fas$ or $\fab$, but not both. 
We show that there exist $2(t - 1)$ such checks, which are used in the Download
Phase. Each RN uses $t-1$
checks and downloads the corresponding $n-2$ sub-symbols from each available node. 
Apart from the $t-1$ checks that involve $\fas$ but not $\fab$, and the $t-1$ checks that involve $\fab$ but not $\fas$, we also introduce
\emph{two} additional checks that involve both $\fas$ and $\fab$, which are useful
in the Collaboration Phase. 
It is not immediately clear how these last two checks can be used at all. 
However, we prove that when the extension degree $t$ is divisible by the characteristic of the field $F$, each erased symbol can be recovered at each
RN using the aforementioned $t$ checks, at the cost of downloading in
total $n-1$ sub-symbols from $n-2$ surviving nodes and from the other RN. In the first repair scheme, the two RNs exchange their repair data simultaneously (parallel repair), while in the second scheme, one node \emph{waits} to receive the data from the other node before sending out its own repair data (sequential repair). By allowing one node to wait in the Collaboration Phase, we obtain a repair scheme that works for every field extension degree. 

To identify check equations that involve one codeword symbol $f(\alpha)$ 
but not the other symbol $f(\beta)$, we first introduce a special polynomial $\Qabz$, defined as follows: %\vspace{-5pt}
\begin{equation}
\label{eq:Q}
\Qabz = \tr\big(z(\beta - \alpha)\big), \quad \alpha \neq \beta.
\end{equation} 
Let $\Kab$ denote the root space of $\Qabz$.
Then %\vspace{-5pt}
\begin{equation} 
\label{eq:Kab}
\Kab = \left\{z^* \in F \colon \tr(z^*\alpha) = \tr(z^*\beta)\right\}.
\end{equation}

\begin{lemma} 
\label{lem:Kab}
The following statements hold for every $\alpha$ and $\beta$ in $F$, $\alpha \neq \beta$.  
\begin{enumerate}
	\item[(a)] $\Kab \equiv \Kba$. In other words, the polynomial $\Qab$ and 
	the polynomial $\Qba$ have the same root spaces.
	\item[(b)] $\dim_B(\Kab) = \dim_B(\Kba) = t - 1$. 
\end{enumerate}
\end{lemma} 
\begin{proof} 
From \eqref{eq:Kab}, due to symmetry, $\Kab \equiv \Kba$.
As the trace function is a linear mapping from $F$ to $B$, its kernel $K = \{\kappa \in F \colon \tr(\kappa)=0\}$ is a subspace of dimension $t-1$ over $B$ (see~\cite[Thm.~2.23]{LidlNiederreiter1986}). Therefore, the root space of $\Qabz$ is $\Kab = \frac{1}{\beta-\alpha}K$, which is also a subspace of dimension $t-1$ over $B$. 
\end{proof}

We then use a root $z^*$ of the polynomial $\Qabz$ to define a check equation according to \eqref{eq:p}. %\vspace{-5pt}
\[
p_{z^*,\alpha}(x) = \tr\big(z^*(x - \alpha) \big)/(x - \alpha). 
\]
The following properties of $p_{z^*,\alpha}(x)$ will be used in our subsequent proofs.

\begin{lemma} 
\label{lem:Qp}
Suppose that $\alpha$ and $\beta$ are two distinct elements of $F$,
and $z^*$ is a root of $\Qabz$ or $\Qbaz$ in $F$, i.e. $z^* \in \Kab$. 
Then the following claim holds.
\begin{enumerate}
	\item[(a)] $p_{z^*,\alpha}(\beta) = 0$. 
\end{enumerate}
Moreover, if the extension degree $t$ is divisible by \text{char}$(F)$ then 
\begin{enumerate}
	\item[(b)] $p_{u,\alpha}(\beta)$ is a root of $\Qabz$ and $\Qbaz$, for every $u \in F$.  
\end{enumerate}
\end{lemma}
\begin{proof} 
Note that according to Lemma~\ref{lem:Kab}~(a), the root spaces of $\Qabz$ and $\Qbaz$
are the same. The first claim is clear based on the definitions of $\Qabz$ and $p_{z^*,\alpha}(x)$.
For the second claim, it is sufficient to show that $p_{u,\alpha}(\beta)$ is a root of $\Qabz$.
 
For simplicity, let $\Delta \define \beta-\alpha$ and $b \define \tr\big(u(\beta 
- \alpha)\big) \in B$. By definition of $p_{u,\alpha}(x)$, we have
$p_{u,\alpha}(\beta) = \tr\big(u(\beta - \alpha)\big) / (\beta - \alpha)
=b / \Delta$.
By definition of $\Qabz$, we also have 
$\Qabz = \tr\big(z(\beta - \alpha) \big)
= \tr(z\Delta)$. 
Therefore,
$\Qab\big(p_{u,\alpha}(\beta)\big) = \tr\Big((b/\Delta)\Delta\Big) 
= \tr(b) = 0$,
because for $b \in B$, we always have
%$\tr(b) = b\tr(1) = b\sum_{i = 0}^{t-1} 1 = tb = 0$,
$\tr(b) = tb = 0$,
whenever $t$ is divisible by the char$(F)$. 
Hence, $p_{u,\alpha}(\beta)$ is a root of $\Qabz$. 
\end{proof} 

The following lemma restates what is shown in Section~\ref{subsec:GW}.  

\begin{lemma} 
\label{lem:trace}
For $\alpha \neq \alpha^*$ and $u \in F$, \vspace{-5pt}
\begin{equation} 
\label{eq:trace}
\tr\big( p_{u,\alpha^*}(\alpha)\fa\big) 
= \tr\big( u(\alpha-\alpha^*) \big)\tr\Big(\dfrac{\fa}{\alpha - \alpha^*} \Big). \vspace{-5pt}
\end{equation}
Hence, $\tr\big( p_{u,\alpha^*}(\alpha)\fa\big)$ can be computed by downloading the repair trace $\tr\Big(\frac{\fa}{\alpha - \alpha^*} \Big)$ from the node storing $\fa$. %This sub-symbol does not depend on $u$.  
\end{lemma} 

\subsection{A Depth-One Repair Scheme for Two Erasures}
\label{subsec:depth_one}

The scheme comprises of two phases, the Download Phase, where each RN contacts and downloads data from the other $n-2$ available nodes, and the 
Collaboration Phase, where the two RNs exchange the data, based on what
they receive earlier in the Download Phase. The main task is to design
the data to be exchanged during the two phases. This task can be completed via 
a selection of proper check polynomials to be used by each RN. 
We discuss the generation of these polynomials below. 

Let $\Ksb$ be the root space of the polynomial $\qsbz$. By Lemma~\ref{lem:Kab}~(b), $\dim_B(\Ksb) = t-1$. Let $U = \{u_1,u_2,\ldots,u_{t-1}\} \subseteq F$ and $V = \{v_1,v_2,\ldots,v_{t-1}\} \subseteq F$ be two arbitrary bases of $\Ksb$ over $B$.
We extend $U$ and $V$ to obtain the two bases $U' = \{u_1,\ldots,u_t\}$ and $V' = \{v_1,\ldots,v_t\}$ of $F$ over $B$, respectively.
For $i \in [t]$, we set \vspace{-5pt}
\begin{equation} 
\label{eq:pix}
p_i(x) \define p_{u_i,\alpha^*}(x) = \tr\big(u_i(x - \alpha^*) \big) / (x - \alpha^*),\vspace{-5pt}
\end{equation} 
\begin{equation} 
\label{eq:qix}
q_i(x) \define p_{v_i,\overline{\alpha}}(x) = \tr\big(v_i(x - \overline{\alpha}) \big)(x - \overline{\alpha}).  \vspace{-5pt}
\end{equation} 

\textbf{Download Phase.}
In this phase, each RN contacts $n-2$ available nodes to download repair
data. To determine what to download, the RN for $\fas$ uses the
first $t-1$ checks $p_1,\ldots,p_{t-1}$ to construct the following $t-1$ repair equations. %For $i = 1,\ldots,t-1$, \vspace{-5pt}
\begin{equation}
\label{eq:pi} 
\tr\big(\pias\fas\big) = -\mkern-18mu\sum_{\alpha \in A \setminus \{\alpha^*\}} \tr\big(\pia\fa\big),\ i \in [t-1].
\end{equation} 
Similarly, the RN for $\fab$ creates the following repair equations. \vspace{-5pt} 
%for $t = 1,\ldots, t-1$.  \vspace{-5pt}
\begin{equation} 
\label{eq:qi}
\tr\big(\qiab\fab\big) = -\mkern-18mu\sum_{\alpha \in A \setminus \{\overline{\alpha}\}} \tr\big(\qia\fa\big),\ i \in [t-1].
\end{equation} 
By Lemma~\ref{lem:Qp}~(a), we have $\piab = 0$ and $\qias = 0$ for
all $i = 1,\ldots,t-1$. Therefore, the right-hand sides of 
\eqref{eq:pi} and \eqref{eq:qi} do not involve $\fas$ and $\fab$. 
As a result, each RN can recover $t-1$ independent traces of the corresponding
erased symbol by downloading $n-2$ sub-symbols (traces) from the available nodes. Corollary~\ref{cr:downloading_phase}, which follows directly from Lemma~\ref{lem:trace}, formally states this
fact. 

\begin{corollary} 
\label{cr:downloading_phase}
In the Download Phase, the replacement node for $\fas$ can recover $t-1$ independent traces, namely
$\tr\big(\poas\fas \big),\ldots,\tr\big(\ptmoas\fas \big)$, by downloading $n-2$
repair traces, i.e. $\tr\Big(\frac{\fa}{\alpha - \alpha^*} \Big)$ from the available
node storing $\fa$, for all $\alpha \in A \setminus \{\alpha^*, \overline{\alpha}\}$. A similar statement holds for the replacement node for $\fab$, where the checks
are $q_i$ and the repair traces are $\tr\Big(\frac{\fa}{\alpha - \overline{\alpha}} \Big)$. 
\end{corollary} 

\textbf{Collaboration Phase.}
As one more independent trace of each erased symbol is needed for a complete recovery, the two RNs create two additional repair equations for $\fas$, $\fab$, respectively. \vspace{-10pt}
\begin{multline}
\label{eq:pt}
\tr\big(\ptas\fas\big) + \tr\big(\ptab\fab\big)\\ 
= -\sum_{\alpha \in A \setminus \{\alpha^*,\overline{\alpha}\}} \tr\big(\pta\fa\big).
\end{multline}
\vspace{-15pt}
\begin{multline}
\label{eq:qt}
\tr\big(\qtab\fab\big) + \tr\big(\qtas\fas\big)\\ 
= - \sum_{\alpha \in A \setminus \{\alpha^*,\overline{\alpha}\}} \tr\big(\qta\fa\big).
 \vspace{-10pt}
\end{multline}
It is clear that from the repair traces $\tr\Big(\frac{\fa}{\alpha - \alpha^*} \Big)$, $\alpha \in A \setminus \{\alpha^*, \overline{\alpha}\}$, retrieved in the
Download Phase, the RHS of~\eqref{eq:pt} can be determined. 
However, to determine the desired trace $\tr\big(\ptas\fas\big)$, the RN for $\fas$ needs to know the missing trace $\tr\big(\ptab\fab\big)$, which would have been downloaded from the node storing $\fab$ if it had not failed.
The following lemma states that for certain field extension degrees, this missing piece of information can be created by the RN for $\fab$ based on what it obtains in the Download Phase. It can then send this trace to the RN for $\fas$ to help complete the recovery of that symbol. A similar scenario also holds for $\fab$.    

\begin{lemma} 
\label{lem:dependence}
If the field expansion degree $t$ is divisible by the characteristic of the fields $F$ and $B$, 
then $\ptab$ is dependent (over $B$) on the set $\{\qiab \colon i \in [t-1]\}$. Also, in this case, 
$\qtas$ is dependent (over $B$) on the set $\{\pias \colon i \in [t-1]\}$. 
\end{lemma} 
\begin{proof} 
Because of symmetry, it suffices to just prove the first statement of the lemma. 
By Lemma~\ref{lem:p}~(b), we have $\qiab = v_i$, for every $i \in [t-1]$. Therefore, 
$\{\qiab \colon i \in [t-1]\} = \{v_1,\ldots,v_{t-1}\} = V$, which is a basis of the
root space $\Ksb$ of the polynomial $\qsbz$. 
Therefore, in order to show that $\ptab$ is dependent on $V$, it is sufficient to prove that $\ptab$ is a root of $\qsbz$.
But this follows immediately from Lemma~\ref{lem:Qp}~(b), because $p_t(x)$ equals $p_{u_t,\alpha^*}(x)$ by its definition in Step~4.   
\end{proof} 

From the linearity of the trace function, we arrive at the following corollary of Lemma~\ref{lem:dependence}. 

\begin{corollary}
\label{cr:dependence}
If the field expansion degree $t$ is divisible by the characteristic of the fields $F$ and $B$, 
then the trace $\tr\big(\ptab\fab\big)$ can be written as a linear combination (over $B$) of the traces in
$\Big\{\tr\big(\qiab\fab \big) \colon i \in [t-1]\Big\}$. Also, the trace $\tr\big(\qtas\fas\big)$ can be
written as a linear combination (over $B$) of the traces in $\Big\{\tr\big(\pias\fas\big) \colon i \in [t-1]\Big\}$. 
Moreover, the coefficients of these combinations do not depend on $f$. 
\end{corollary} 

Note that by Lemma~\ref{lem:trace}, the traces $\tr\big(\ptab\fab\big)$ and $\tr\big(\qtas\fas\big)$ can be determined based on the traces $\tr\Big(\frac{\fab}{\overline{\alpha} - \alpha^*} \Big)$
and $\tr\Big(\frac{\fas}{\alpha^* - \overline{\alpha}} \Big)$, respectively. Therefore, in the Collaboration Phase, the RNs can send their repair data to each other, which matches precisely what they would have sent if they had not failed.
The graphical illustration of the two phases of this scheme is depicted in Fig.~\ref{fig:distributed_two}. We refer to this as a \emph{depth-one} collaborative repair scheme because in the Collaboration Phase, two RNs exchange repair data in one round and do not have to wait for each other.  

\vspace{-5pt}
\begin{figure}[htb]
\centering
\includegraphics[scale=.8]{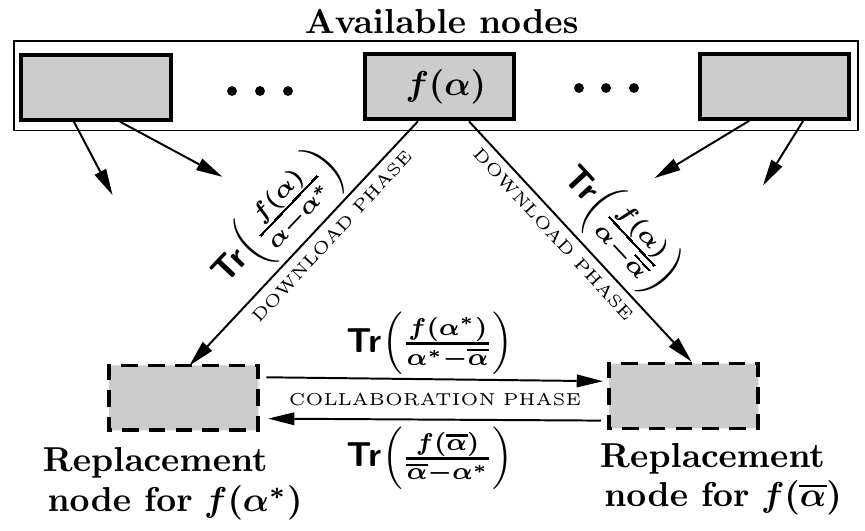}
\caption{Illustration of the depth-one collaborative repair scheme of two erasures for Reed-Solomon codes.}
\label{fig:distributed_two}
\vspace{-15pt}
\end{figure}

\begin{lemma} 
\label{lem:collaboration_phase}
In the Collaboration Phase, the replacement node for $\fas$ can recover the $t$-th trace $\tr\big(\ptas\fas \big)$, by downloading one repair trace $\tr\Big(\frac{\fab}{\overline{\alpha} - \alpha^*} \Big)$ from the replacement node for $\fab$. Similarly, the replacement node for $\fab$ can recover the $t$-th trace $\tr\Big(\frac{\fa}{\alpha - \overline{\alpha}} \Big)$ by downloading one repair trace $\tr\Big(\frac{\fas}{\alpha^* - \overline{\alpha}} \Big)$ from the replacement node for $\fas$. 
\end{lemma} 

\begin{theorem} 
\label{thm:depth_one}
The depth-one collaborative repair scheme can be used to repair any two erased symbols of a Reed-Solomon codes $\rsk$ at a repair bandwidth of $n-1$ sub-symbols per symbol, given that $n-k \geq |B|^{t-1}$ and
the characteristic of $F$ divides $t$.% $\text{char}(F) | t$.  
\end{theorem}  
\begin{proof}\noindent 
By Lemma~\ref{lem:p}~(b), $\pias = u_i$ and $\qiab = v_i$ for $i \in [t]$.
Recall that the sets $U' = \{u_1,\ldots,u_t\}$ and $V' = \{v_1,\ldots,v_t\}$ are both  linearly independent over $B$. Therefore, after the two phases, each RN obtains $t$ independent traces for each erased symbol, $\tr\big(u_i\fas\big)$, for $\fas$, and $\tr\big(v_i\fab\big)$, for $\fab$, for all $i \in [t]$.  
Thus, each erased symbol will have $t$ independent traces for its recovery.
Each RN downloads $n-2$ sub-symbols in the Download Phase and one 
sub-symbol in the Collaboration Phase, according to Corollary~\ref{cr:downloading_phase}
and Lemma~\ref{lem:collaboration_phase}, which sum up to a total repair bandwidth of $n-1$ sub-symbols. 
\end{proof} 

\begin{remark}
\label{rm:distributedII}
In our repair scheme, each RN uses a bandwidth of $n-1$ sub-symbols, which is the same as the case of one erasure in~\cite{GuruswamiWootters2016}. In a naive scheme, one RN first downloads $kt$ sub-symbols from some $k$ available nodes, recover both erased symbols, and then sends the corresponding symbol to the other RN. Its total bandwidth used is $kt+t$, which is worse than ours if $\frac{k+1}{n-1} \geq \frac{2}{t}$, i.e. when $t$ is not small or the code has high rate. 
\end{remark}

\begin{example}
\label{ex:2}
Let $q = 2$, $t = 2$, $n = 4$, and $k = 2$. 
Let $\ff_4 = \{0,1,\xi,\xi^2\}$, where $\xi^2+\xi+1 = 0$. Then $\{1,\xi\}$ is a
basis of $F=\ff_4$ over $B=\ff_2$. Moreover, each element $\ba \in \ff_4$ can be represented
by a pair of bits $(a_1,a_2)$ where $\ba = a_1 + a_2\xi$. Suppose the stored file is
$(\ba,\bb) \in \ff_4^2$. To devise a systematic RS code, 
we associate with each file $(\ba,\bb) \in \ff_4^2$ a polynomial $f(x)=f_{\ba,\bb}(x)\define \ba + (\bb-\ba)x$.
%, which is of degree at most $1 = k - 1$. 
%The evaluations of this polynomial at the four elements of the field $\ff_4$ are given as follows.
We have 
\[
\begin{split}
f(0) &= a_1 + a_2\xi = \ba,\\
f(1) &= b_1 + b_2\xi = \bb,\\
f(\xi) &= (a_1+a_2+b_2) + (a_1+b_1+b_2)\xi,\\
f(\xi^2) &= (a_2+b_1+b_2) + (a_1+a_2+b_1)\xi.
\end{split}
\]
The four codeword symbols $f(0)$, $f(1)$, $f(\xi)$, and $f(\xi^2)$ are stored at Node~1,
Node~2, Node~3, and Node~4, respectively, as depicted in 
%Fig.~\ref{fig:distributed_two_example}. 
Fig.~\ref{fig:toy_example}.

\textbf{Download Phase.}
Set 
\[
Q_{1,\xi}(z) \define \tr\big( z(\xi-1)\big) = \xi z^2 + \xi^2z.
\]
We choose two bases $U = V = \{\xi\}$ of the root space of $Q_{1,\xi}(z)$. 
Set $p_1(x) = \tr\big(\xi(x-1)\big) / (x-1) = \xi^2x + 1$, and 
$q_1(x) = \tr\big(\xi(x-\xi)\big) / (x-\xi) = \xi^2x+\xi^2$. 
RN2 (RN for Node~2) downloads two bits from the two available nodes, namely $a_2 = \tr\big(f(0) / (0-1)\big)$ from Node~1 and
$a_2+b_1+b_2 = \tr\big(f(\xi^2) / (\xi^2-1)\big)$ from Node~4. 
It then uses \eqref{eq:pi} to obtain the first trace
$b_1+b_2 = \tr(\xi f(1)) = 1\times a_2 + 1\times(a_2+b_1+b_2)$. 
Similarly, RN3 (RN for Node~3) also downloads 
$a_1 = \tr\big(f(0) / (0-\xi)\big)$ from Node~1 and
$a_1+a_2+b_1 = \tr\big(f(\xi^2) / (\xi^2-\xi)\big)$ from Node~4.
It then recovers $a_2+b_1 = \tr(\xi f(\xi)) = 1\times a_1 + 1\times (a_1+a_2+b_1)$.

\textbf{Collaboration Phase.} 
RN2 sends $b_1+b_2$ over to the RN3, which, by Lemma~\ref{lem:dependence}, is the same as $\tr\big(f(1)/(1-\xi)\big)$. Conversely, RN3 sends $a_2+b_1$ over to
RN2, which is the same as $\tr\big(f(\xi)/(\xi-1)\big)$. 
$U$ and $V$ are extended to the basis $\{\xi,\xi^2\}$ of $\ff_4$ over $\ff_2$.
Set $p_2(x) = \tr\big(\xi^2(x-1)\big) / (x-1) = \xi x + 1$, and 
$q_2(x) = \tr\big(\xi^2(x-\xi)\big) / (x-\xi) = \xi x$.
Now, RN2 has three repair traces to recover
the second trace of $f(1)$ using $p_2$, i.e. $b_1 = \tr\big(\xi^2f(1)\big) = 
1\times a_2 + 0\times(a_2+b_1+b_2) + 1\times (a_2 + b_1)$.
Based on the two traces $b_1+b_2$ and $b_1$, the erased symbol $\bb=f(1)$ can be recovered. 
Similarly, RN3 can recover the second trace of $f(\xi)$
as $a_1+a_2+b_2 = \tr\big(\xi^2f(\xi)\big) = 0\times(a_1) + 1\times(b_1+b_2) + 1\times(a_1+a_2+b_1)$, and then can recover $f(\xi)$ completely. 
\end{example}    

\subsection{A Depth-Two Repair Scheme for Two Erasures}

We modify the depth-one repair scheme developed in the previous subsection to obtain
a depth-two scheme that works for \emph{all} field extension degrees. 
We still generate the checks $p_1,\ldots,p_t$ and $q_1,\ldots,q_t$ as in the first scheme,
given by \eqref{eq:pix} and \eqref{eq:qix}, respectively. However, the RN for $\fas$,
instead of $p_i$, uses the following checks %\vspace{-5pt}
\begin{equation} 
\label{eq:ppi}
\psix \define \tau p_i(x),\quad \text{ for all } i=1,\ldots,t. %\vspace{-5pt}
\end{equation} 
where $\tau \define \qoab / \ptab = v_1 / \ptab$ is a \emph{nonzero} constant, which only depends on $\alpha^*$, $\overline{\alpha}$, $u_t$, and $v_1$, and not on $f$. 
Note that since $u_t \notin \Ksb$, $\ptab \neq 0$. Hence, $\tau$ is well defined. 
Clearly, $\deg(\psi)=\deg(p_i) < n-k$ and hence, $\psix$ serve as check polynomials of the code.
 
\begin{lemma} 
\label{lem:ps}
The check polynomials $\psix$ defined as in \eqref{eq:ppi} satisfy the following 
properties. 
\begin{itemize}
	\item[(P1)] $\psoab = \pstab = \cdots =  \pstmoab = 0$.
	\item[(P2)] $\pstab = \qoab = v_1 \neq 0$. 
	\item[(P3)] $\{\psoas,\ldots,\pstas\}$ is a basis of $F$ over $B$. 
\end{itemize}
\end{lemma}  
\begin{proof} 
The first property (P1) holds because $p^*_i(\overline{\alpha}) = \tau \piab$ and $\piab = 0$ for every $i=1,\ldots,t-1$ by Lemma~\ref{lem:Qp}~(a).  
Property (P2) is obvious.   
Property (P3) follows from the fact that $\psias = \tau\pias$, $\tau \neq 0$, 
and that $U' = \{u_1,\ldots,u_t\} = \{\poas,\ldots,\ptas\}$ is a basis of $F$ over $B$.  
\end{proof} 

\textbf{Download Phase.}
In this phase, the RN for $\fas$ uses the first $t-1$ checks $p^*_1,\ldots,p^*_{t-1}$ to construct the $t-1$ repair equations. For $i = 1,\ldots,t-1$, \vspace{-5pt}
\begin{multline}
\label{eq:psi} 
\tr\big(\psias\fas\big) = -\sum_{\alpha \in A \setminus \{\alpha^*\}} \tr\big(\psia\fa\big)\\
=-\sum_{\alpha \in A \setminus \{\alpha^*\}} 
\tr\big(u_i(\alpha-\alpha^*)\big)\tr\Big(\frac{\tau\fa}{\alpha-\alpha^*}\Big). \vspace{-5pt}
\end{multline}
By Lemma~\ref{lem:ps}~(a), the RHS of \eqref{eq:psi} does not involve $\fab$. Thus, the RN for $\fas$ can determine $t-1$ traces $\tr\big(\psias\fas\big)$, $i \in [t-1]$, of $\fas$
by downloading $n-2$ sub-symbols $\tr\Big(\frac{\tau\fa}{\alpha-\alpha^*}\Big)$ from
the available nodes storing $\fa$, $\alpha \in A \setminus \{\alpha^*, \overline{\alpha}\}$. 
The RN for $\fab$ follows the same procedure as in the first scheme (Section~\ref{subsec:depth_one}).

\textbf{Collaboration Phase.}
The last repair equation for $\fas$ is \vspace{-8pt} 
\begin{multline}
\label{eq:pst}
\tr\big(\pstas\fas\big) + \tr\big(\pstab\fab\big)
=-\mkern-18mu\sum_{\alpha \in A \setminus \{\alpha^*,\overline{\alpha}\}}\mkern-18mu \tr\big(\psta\fa\big)\\ 
= -\sum_{\alpha \in A \setminus \{\alpha^*,\overline{\alpha}\}} 
\tr\big(u_t(\alpha-\alpha^*)\big)\tr\Big(\frac{\tau\fa}{\alpha-\alpha^*}\Big). %\vspace{-3pt}
\end{multline}

Clearly, the RN for $\fas$ can compute the RHS of \eqref{eq:pst} based on what
it downloaded in the Download Phase. To obtain the last trace $\tr\big(\pstas\fas\big)$, it downloads $\tr\big(\pstab\fab\big)$ from the RN for $\fab$,
which is possible because $\tr\big(\pstab\fab\big) = \tr\big(\qoab\fab\big)$,
due to Lemma~\ref{lem:ps}~(P2), which is already available at the RN for $\fab$. 
Then, due to Lemma~\ref{lem:ps}~(P3), the RN for $\fas$ has $t$ independent
traces of $\fas$ to recover this lost symbol. As $\fas$ has been recovered, 
the RN for $\fab$ downloads the repair trace $\tr\big(\fas/(\alpha^* -\overline{\alpha} )\big)$ from the RN for $\fas$ to compute the trace $\tr\big(\qtab\fab\big)$,
and then can recover $\fab$ completely. 
Note that the RN for $\fas$ has to first receive the repair trace from the RN for $\fab$ before computing and sending out its repair trace for $\fab$. 
Theorem~\ref{thm:depth_two} summarizes the discussion. \vspace{-3pt}  

\begin{theorem} 
\label{thm:depth_two}
The depth-two collaborative repair scheme can be used to repair any two erased symbols of a Reed-Solomon codes $\rsk$ at a repair bandwidth of $n-1$ sub-symbols per symbol, for every field extension degree, given that $n-k \geq |B|^{t-1}$.  
\end{theorem} \vspace{-5pt}

%It is showed very recently in~\cite{DauMilenkovic2017} that a more general way to  generate the check polynomials used in a repair scheme is via linearized polynomials whose root sets consist of elements of a subspace. 
%We remark that by selecting a subspace of dimension $t-1$ that contains $1$, we obtain a depth-one repair scheme for every value of $t$.
%Our depth-two scheme, however, provides a way to perform a much faster heuristic search for a low-bandwidth repair scheme for codes with arbitrary number of parities, such as RS(9,6) or RS(14,10).    
%We omit the details, due to lack of space.  

\section*{Acknowledgment}

This work has been supported in part by the NSF grant CCF 1526875, the Center for Science of Information under the grant NSF 0939370, and the NSF grant CCF 1619189. The work for H. M. Kiah has been supported in part by a Singapore Ministry of Education (MOE) Tier 1 grant 2016-T1-001-156.

\bibliographystyle{IEEEtran}
\bibliography{RepairingRSCodes_MultipleErasures}

\end{document}